\documentclass[12pt]{article}
\usepackage{a4wide}
\usepackage[utf8]{inputenc}
\usepackage{authblk}
\usepackage{lipsum}

\title{Cryptanalysis of some Nonabelian \\ Group-Based Key Exchange Protocols \thanks{This research is supported by Armasuisse Science and Technology. The third author has also been supported by the Swiss National Science Foundation under grant number 212865.}}

\date{}

\author[1]{Simran Tinani}
\author[2]{Carlo Matteotti}
\author[1]{Joachim Rosenthal}
\affil[1]{\small Institute of Mathematics, University of Zurich, Winterthurerstrasse 190,  8057 Zurich} 

\affil[2]{\small Eidgenössisches Departement VBS, \small Papiermühlestrasse 20, 3003 Bern}

\usepackage{graphicx}
\usepackage{amsmath, amssymb}
\usepackage[utf8]{inputenc}
\usepackage{subcaption}
\usepackage{tabularx}
\usepackage{amsmath,amssymb,amsthm}
\usepackage[english]{babel}
\usepackage{mathtools}
\usepackage{enumerate}
\usepackage[numbers]{natbib}
\usepackage[dvipsnames]{xcolor}
\usepackage{color}
\usepackage{comment}

\usepackage{hyperref}
\hypersetup{
  colorlinks   = true, 
  urlcolor     = blue, 
  linkcolor    = OliveGreen, 
  citecolor   = red 
}

\usepackage{tabularx}
\usepackage{booktabs}
\usepackage{makecell}
\usepackage{longtable}
\usepackage[ruled,vlined]{algorithm2e}
\usepackage{algorithmic}
\usepackage{listings}
\usepackage{etoolbox}
\usepackage{lipsum}

\usepackage{url}

\usepackage{amssymb}
\usepackage{amsfonts}
\usepackage{amsmath}
\usepackage{comment}
\usepackage{tabularx}
\usepackage{booktabs}
\usepackage{makecell}
\usepackage{longtable}
\usepackage{graphicx}
\usepackage{caption}
\usepackage{lipsum}
\usepackage{xcolor}
\usepackage{parskip}
\usepackage{caption}
\usepackage{floatrow}

\usepackage{hyperref, enumitem}
\DeclareMathOperator{\Aut}{Aut}
\DeclareMathOperator{\Hol}{Hol}
\DeclareMathOperator{\Mat}{Mat}

\DeclareMathOperator{\ord}{ord}

\DeclareMathOperator{\GL}{GL}

\DeclareMathOperator{\lcm}{lcm}
\newcommand{\Z}{\mathbb{Z}}

\newcommand{\F}{\mathbb{F}}

\DeclarePairedDelimiter\card{\lvert}{\rvert}
\newfloatcommand{capbtabbox}{table}[][\FBwidth]

\newcommand{\fq}{{\mathbb{F}_q}}

\newtheorem{theorem}{Theorem}
\newtheorem{protocol}{Protocol}

\newtheorem{lemma}{Lemma}
\newtheorem{proposition}{Proposition}
\newtheorem{corollary}{Corollary}
\theoremstyle{definition}

\newtheorem{definition}{Definition}
\theoremstyle{remark}
\newtheorem{remark}{Remark}

\providecommand{\keywords}[1]{\textbf{\textit{Keywords---}} #1}

\begin{document}
\maketitle
 \begin{abstract}
In the recently emerging field of nonabelian group-based cryptography, a prominently used one-way function is the Conjugacy Search Problem (CSP), and two important classes of platform groups are polycyclic and matrix groups. In this paper, we discuss the complexity of the conjugacy search problem (CSP) in these two classes of platform groups using the three protocols in \cite{Gu2014ConjugacySB}, \cite{cmpdec}, and \cite{quat} as our starting point.  We produce a polynomial time solution for the CSP in a finite polycyclic group with two generators, and show that a restricted CSP is reducible to a DLP. In matrix groups over finite fields, we usedthe Jordan decomposition of a matrix to produce a polynomial time reduction of an $A$-restricted CSP, where $A \subseteq GL(\fq)$ is a cyclic subgroup, to a set of DLPs over an extension of $\fq$. We use these general methods and results to describe concrete cryptanalysis algorithms for these three systems. In particular, we show that in the group of invertible matrices over finite fields and in polycyclic groups with two generators, a CSP where conjugators are restricted to a cyclic subgroup is reducible to a set of $\mathcal{O}(n^2)$ discrete logarithm problems. Using our general results, we demonstrate concrete cryptanalysis algorithms for each of these three schemes.
 We believe that our methods and findings are likely to allow for several other heuristic attacks in the general case.

    \end{abstract}

\\

 \section{Introduction}

The construction and realization of cryptographic systems that resist quantum attacks presently constitutes an important area of research. Apart from lattice-based, multivariate, isogeny-based and code-based cryptography, it has been proposed recently to use the rich structure of nonabelian groups to construct quantum-secure protocols for public key exchange, message encryption, and authentication. Some recent surveys on this emerging field, called group-based cryptography, can be found in Fine et al. \cite{fine2011aspects} and Myasnikov et al. \cite{myasnikov2008group}.

 The most prominent algorithmic problem employed for constructing nonabelian protocols is the Conjugacy Search Problem (CSP). While the Discrete Logarithm Problem (DLP) in a group $G$ requires the recovery of the exponent $n$ when given the group elements $g$ and $h=g^n$, the CSP requires the recovery of a conjugator $x \in G$, given the elements $g$ and $h=x^{-1} g x$. To reflect this analogy, it is common to use the notation $g^x:= x^{-1}g x$ for $g, x \in G$, which we also adopt in this paper. If the conjugator is restricted to lie in a subgroup $A\subseteq G$, we refer to the problem as an $A$-restricted CSP. In this paper we are specifically interested in the case where $A$ is cyclic. We remark that conjugation is an action of a group on itself, and thus all CSP-based protocols may be seen as special cases of the semigroup action-based framework introduced by Maze et al. in \cite{maze}.
 
The first and most prominently known protocols constructed based on the CSP were by Anshel, Anshel and Goldfeld (AAG) \cite{an99}, and Ko-Lee \cite{ko2000new}, both of whose underlying problems is a specific restricted CSP. The authors of both these systems proposed as platforms the Braid groups $B_N$. However, a number of attacks \cite{braid1}, \cite{braid2}, \cite{tsaban2015polynomial} show that the braid groups are not suitable platforms. Nevertheless, the possibility of finding another potential nonabelian platform group for CSP-based protocols is still open to research. Some other groups that have been proposed for use are polycyclic groups, metabelian groups, $p$-groups, Thompson groups, and matrix groups.  

When the underlying platform group is linear (i.e. embeds faithfully into a matrix group over a field), several polynomial time attacks exist, which focus on retrieving the private shared key without solving the CSP \cite{kreuzer2014linear}, \cite{myasnikov2015linear},  \cite{tsaban2015polynomial}, \cite{ben2018cryptanalysis}. However, in general the computation of an efficient linear representation may pose a serious roadblock for an adversary.  Further, many of these attacks are impractical to implement for standard parameter values. Such attacks are also typically protocol-specific, and this always leaves open the possibility of constructing a different protocol, again based on the CSP, where the known attacks are avoided. So far, the true difficulty of the CSP in different platforms has not been sufficiently investigated.

  In this paper, we study the complexity of various versions of the CSP in two well-known classes of linear groups: polycyclic groups and matrix groups over finite fields. Polycyclic groups were suggested for cryptographic use in \cite{eick2004polycyclic} by Eick and Kahrobaei, where some evidence was provided for resistance to some known attacks. Matrix groups are important for any nonabelian system, since whenever the platform group is linear, an efficient faithful representation reduces the underlying problem to one in a matrix group. Several proposed nonabelian cryptosystems use platforms that are special instances of polycyclic or matrix groups, and employ problems either equivalent to, or easier than, the versions of the CSP discussed in this paper. In this paper, we highlight three such independent key exchange systems published in \cite{Gu2014ConjugacySB}, \cite{cmpdec}, and \cite{quat}, and demonstrate efficient cryptanalyses of each of these.

 The paper is organized as follows. We start in Section~\ref{sec: protocols} by describing the three nonabelian group-based protocols in \cite{Gu2014ConjugacySB}, \cite{cmpdec}, and \cite{quat}, and their underlying structures and algorithmic problems. In Section~\ref{pc-sec}, we deal with polycyclic groups, giving a brief overview on their structure, discuss the algorithms and complexities for the group operations, and further discuss the complexity of the CSP in some special cases. In Subsection~\ref{pc-csp} we show that in the case with two generators, there is a polynomial time solution for the CSP, and that the $A$-restricted CSP, for a suitable choice of a cyclic subgroup $A$, is equivalent to a DLP. We also show some reductions for polycyclic groups with $3$ generators, and for the general $n$ generator case where $n-1$ generators commute. We demonstrate two examples, also showing an extension of the method to solving decomposition problems. In Subsection~\ref{pc-mat}, we discuss the consequences if a matrix representation is used in place of the standard polycyclic presentation. In Section~\ref{mat-sec}, we deal with matrix groups, and provide a reduction of the $\langle Z \rangle $-restricted CSP in $\GL_n(\fq)$ to $\mathcal{O}(n^2)$ DLPs, dealing separately with the cases where $Z$ is diagonalizable and not. Finally, in Section~\ref{sec-cryptanalysis} we demonstrate a cryptanalysis of each of the cryptosystems described in Section~\ref{sec: protocols} using the general methods in Sections~\ref{pc-sec} and \ref{mat-sec}. 
 The algebraic reductions demonstrated in this paper may also prove useful in future cryptanalysis techniques for nonabelian protocols over different platform groups.


   \section{Protocols}\label{sec: protocols}

In this section, we describe the three key exchange protocols proposed in \cite{Gu2014ConjugacySB}, \cite{cmpdec}, and \cite{quat}. Each of these schemes is based on a nonabelian linear platform that is either polycyclic or isomorphic to a matrix group. The algorithmic problem underlying the schemes of \cite{Gu2014ConjugacySB} and \cite{quat} is a restricted conjugacy search problem, whereas the problem underlying \cite{cmpdec} is the decomposition problem. The decomposition problem for cryptography was first described in \cite{shpilrain2005new}.

In Sections~\ref{pc-sec} and \ref{mat-sec}, we study these classes of groups in a general setting and discuss the complexity of the conjugacy search problem in some cases.  In Section~\ref{sec-cryptanalysis}, we will describe methods for cryptanalysis for each of these systems, using the general methods in Sections~\ref{pc-sec} and \ref{mat-sec}.


\subsection{Decomposition over Generalized Quaternions }

A generalized quaternion group is a finite polycyclic group given by the presentation
\begin{equation}\label{quatpres}
    Q_{2^n} = \langle x, y \mid x^N = 1, y^2 = x^{N/2}, yx = x^{-1}y, N= 2^{n-1}\rangle. 
\end{equation}

Clearly, any element in this group can be written uniquely as $x^{i}y^j$, where $ 0\leq i \leq N $, $0\leq j \leq 1$. 
In \cite{cmpdec}, Protocol~\ref{compdec1} was proposed in the platform $Q_{2^n}$. This protocol is based on a problem closely related to the decomposition problem, which was introduced in \cite{shpilrain2005new}.


\begin{definition}[Decomposition Problem]
 Let $G$ be a group and $g$ be a known base element. For secret elements $a,b \in G$ and $h:=agb$, recover $a', b' \in G$ such that $h=a'gb'$. 
\end{definition}

Clearly, the conjugacy search problem is a more specific form of decomposition problem, where one always has $b=a^{-1}$. 


\begin{protocol}\label{compdec1}
The public parameters are $G=Q_{2^n}$ given by \eqref{quatpres} and subgroups $A_1, A_2 \subseteq \langle x \rangle $.  \begin{enumerate}\item \begin{enumerate}    \item  Alice picks secret elements $a\in G$, $b_1, b_2 \in A_1$ and sends $x_1 = b_1 a b_2$ to Bob.    \item Bob picks secret elements $d_i \in A_2$ and sends $x_2 = d_1 x_1 d_2$ to Alice.    \item Alice sends $x_3 = b_1^{-1} x_2 b_2^{-1} (= d_1 a d_2)  $ to Bob.\end{enumerate}   \item \begin{enumerate}     \item Bob picks a secret element $c\in G$  and sends $y_1 = d_1 c d_2$ to Alice.    \item Alice sends $y_2 = b_1 y_1 b_2$ to Bob.    \item Bob sends $y_3 =b_1 c b_2 (=d_1^{-1} y_2 d_2^{-1})$ to Alice.    \end{enumerate}\end{enumerate} The shared secret is $b=ac = a(b_1^{-1} y_3 b_2^{-1})= (d_1^{-1} x_3 d_2^{-1})c$.
\end{protocol}

Clearly, an adversary who can solve for the base elements $a$ and $c$ can recover the shared secret. For this, the adversary may solve the decomposition problem for $b_1$ and $b_2$ using the public elements $x_2$ and $x_3$ (or $y_2$ and $y_3$), then subsequently recover $a$ from $x_1$, and then directly compute the shared secret $b$ from the equation above. The security of the protocol therefore relies directly on the decomposition problem in $Q_{2^n}$.

\subsection{Subgroup Conjugacy Search in Quaternions mod $p$}

In \cite{quat}, a protocol for key exchange was proposed using a ring $R=H_p$, which we will refer to as the quaternions mod $p$. We first describe the structure of this ring.

Recall that the set of all Lipchitz quaternions is defined as $L=\{a=a_1+a_2i+a_3j+a_4k : a_1, a_2, a_3, a_4 \in \Z\}.$ Similarly, the Hurwitz quaternions are given by $H=\{a=a_1+a_2i+a_3j+a_4k : a_1, a_2, a_3, a_4 \in \Z+\frac{1}{2} \Z\}.$ 

 For a prime $p$, the authors define $H_p$ as the set $\{a =a_1+a_2i+a_3 j+a_4 k \mid a_i \in \Z_p\}$.  Addition, multiplication, the norm $||\cdot||$, and conjugates are defined in the usual way as in quaternion algebras (for an exposition on arithmetic in quaternion sets, see \cite{hurwitz}), but over the base ring $\Z_p$. The quaternion $a$ is invertible in $H_p$ if and only if $||a||\neq 0 \pmod p$. 
 
 Denote by $H_p^{*}$ the set of all invertible quaternions in $H_p$. The protocol in \cite{quat} is described as follows.
 \begin{protocol}\label{quart}
 \begin{enumerate}
\item  Alice and Bob agree to choose randomly public elements $x\in H_p$ $, z\in H_p^{*}$.
\item Alice picks two secret integers $r,s \in \Z$ such that $1\leq r \leq p-1$ and $2\leq s \leq p-1$ and then computes $y_A=z^r x^s z^{-r}$, and sends $y_A \in H_p$ to Bob.
\item Bob picks two secret integers $u,v \in \Z$ such that $1\leq u \leq p-1$ and $2\leq v \leq p-1$ and then computes $y_B=z^u x^v z^{-u}$, and sends $y_B \in H_p$ to Alice.
\item     Alice computes $K_A=z^r {y_B}^s z^{-r}$ as the shared session key.
\item    Bob computes $K_B=z^u {y_A}^v z^{-u}$ as the shared session key.
\end{enumerate}

 \end{protocol}

Clearly, an adversary who can re-construct any one of the exponent pairs ($r$, $s$) and ($u$, $v$) can recover the shared secret key. Due to the exponentiation also of the base element $x$ during the key computation, this problem is not immediately equivalent to a (restricted) conjugacy search problem.

\subsection{Subgroup Conjugacy Search in Matrix Groups}

In \cite{Gu2014ConjugacySB}, the authors introduce what they call the Subgroup Conjugacy Search Problem (SCSP) in a nonabelian group $G$, and propose a protocol based on it, suggesting as potential platforms the matrix group $\GL_n(\fq)$ and a subgroup of it. We note that the base element is not required to be invertible, and therefore state below a slightly more general form of their protocol.

\begin{protocol}\label{mat-prot}
Let $X \in \Mat_n(\F_q)$ and $Z \in \Mat_n(\F_q)^{*}$ be public elements.
\begin{enumerate}
    \item Alice picks a secret integer $r\in \Z$ and publishes $Z^{-r} X Z^r$.
  \item Bob picks a secret integer $s \in \Z$ and publishes $Z^{-s} X Z^{s}$.
    \item The shared secret is $Z^{-s-r} X Z^{s+r}$.
\end{enumerate}
\end{protocol}

We note that the SCSP problem defined by the authors, i.e., the recovery of the exponent $r$ or $s$ from the public information in the protocol, corresponds exactly to the $A$-restricted CSP defined in this paper, for a cyclic subgroup $A$ of $G$.

\section{Polycyclic groups}\label{pc-sec}

The use of polycyclic groups as a potential platform for conjugacy-based cryptography was 
was first suggested by Eick and Kahrobaei \cite{eick2004polycyclic}.  Some evidence was provided for these groups resisting the length-based attacks afflicting braid group-based systems. Cavallo and Kahrobaei \cite{cavallo} constructed a family of polycyclic groups  where the conjugacy search problem is NP-complete with respect to the respective parameters. Thus, in full generality, there is no known general efficient solution to the conjugacy search problem in polycyclic groups. In fact, there appear to be several classes of polycyclic groups where this problem is infeasible, which may be used to construct various cryptosystems. A survey of polycyclic group-based cryptography can be found in \cite{gryak2016status}.

This section discusses the complexity of some variations of the conjugacy search problem in some classes of polycyclic groups. In polycyclic groups with two generators, we obtain reductions to at most a discrete logarithm problem. We emphasize that these reductions do not diminish the security in the general case. Indeed, the fact that we obtain equivalence to discrete logarithm problems in this simple case, is evidence that the problem may be infeasible for the more complex cases. 

This section also has some results on the complexity of group operations in polycyclic groups.  In particular, we show for two generators that the complexity of multiplications is $\mathcal{O}(1)$. For more generators, when each generator has finite order, we show an upper bound on the complexity. These results on the efficiency of operations in a polycyclic
group are significant in determining its suitability as a platform for cryptography. 

We start by introducing polycyclic groups and some of their basic properties. 

\subsection{Background}\label{pc-basics}
\begin{definition}[Polycyclic Group]
A polycyclic group is a group $G$ with a subnormal series $G = G_1 > G_2 > \ldots > G_{n+1} = 1$ in which every quotient $G_i/G_{i+1}$ is cyclic. This series is called a polycyclic series.
\end{definition}

\begin{definition}[Power-Conjugate Presentation] Let $G$ be a group with generators $a_1, a_2, \ldots, a_n$. Let $I \subseteq \{1,2,\ldots, n\}$ denote a list of indices and $m_i > 1$ be integers corresponding to elements $i \in I$. A power-conjugate presentation is a group presentation of the form \begin{align}\label{pcpres}
    G= \langle a_1, a_2, \ldots, a_n \mid & a_i^{m_i} = w_{ii}, \ i \in I, \nonumber \\
    & a_j^{a_i} = w_{ij}, \ 1 \leq i < j \leq n, \nonumber \\
    & a_j^{a_i^{-1}} = w_{-ij}, \ 1 \leq i < j \leq n, i\not\in I \rangle, \end{align} where the words $w_{ij}$ are of the form $w_{ij} = a_{\mid{i}\mid +1}^{l(i,j, \mid{i}\mid +1)}\ldots a_n^{l(i,j,n)}$, with $l(i,j,k)\in \Z$, and $0 \leq l(i,j,k) < m_k$ if $k\in I$.
\end{definition}

\begin{lemma}[\cite{holt2005handbook}] $G$ is polycyclic if and only if it has a power-conjugate presentation.\end{lemma}

Define $G_i = \langle a_i, a_{i+1} \ldots a_n \rangle $, $1 \leq i < n$, $G_{n+1}=\langle 1\rangle$. The presentation in \eqref{pcpres} is called \emph{consistent} if $\card{ G_i/G_{i+1}} = m_i$ whenever $i \in I$, and the $G_i/G_{i+1}$ is infinite whenever $i \not \in I$.

\begin{definition}[Geodesic Form] Let $G$ be generated by a set of alphabets $X$, and $a_i \in X \cup X^{-1}$, $1\leq i \leq n$. A word $w=a_1 a_2\ldots a_n$ is said to be in geodesic from if there is no shorter word that represents the group element $w$.
\end{definition}

\begin{definition}[Normal Form] Given a consistent polycyclic presentation \eqref{pcpres} for a group $G$, every element $a$ of $G$ can be represented uniquely in the form $a = a_1^{e_1} a_2^{e_2} \ldots a_n^{e_n} $ where $e_i \in \Z$, $0 \leq e_i \leq m_i$ for $i\in I$. This is called the normal form of $a$.
\end{definition}

Henceforth, as is standard, we will use the normal form to represent words in polycyclic groups, and also assume that the presentations we deal with are all consistent. Given a word $w$ in $G$, the process by which minimal non-normal subwords are reduced to normal form using the relations in \eqref{pcpres} is called \emph{collection}. 

Many different strategies for this process have been suggested, but the best-known performance in most cases is achieved by the Collection from the Left Algorithm \cite{VaughanLee1990CollectionFT}, and its improvement in \cite{GEBHARDT2002213}.  Under this, a word $w = (x_1^{\alpha_1}x_2^{\alpha_2}\ldots x_n^{\alpha_n})(x_{i_1}^{\beta_1} \ldots \ldots x_{i_t}^{\beta_{t}})$ is represented in two parts: collected part (as a vector $(\alpha_1, \alpha_2,\ldots, \alpha_n)$) and uncollected stack of generator powers $s=(x_{i_1}^{\beta_1}, \ldots, x_{i_t}^{\beta_t})$. To collect $w$ into normal form, a finite number of generator powers from the uncollected part are iteratively pushed onto the collected part.

  In general, the running time of this algorithm depends on the exponents appearing across the whole process. Due to the dependence on the exponents in the intermediate steps, the complexity remains a rough and unclear estimate. It is clear that collection is a key part of the group operations since it is involved in both multiplication and inversion of words in $G$, and so in general operations in a polycyclic group may not be efficient. However, in many special cases, the complexity can be bounded. We discuss the complexity of computing in polycyclic groups for some cases, including all finite polycyclic groups, in the next subsection. 
  
\subsection{Complexity of Group Operations}\label{pc-compl}

Throughout this section, $G$ denotes a polycylic group given by the presentation $\eqref{pcpres}$. Below, we begin our discussion by showing that collected words in the last two generators can be multiplied, inverted and exponentiated, using explicit formulas. For simplicity, write the relations in $x_{n-1}$ and $x_n$ as $x_n^{x_{n-1}} = x_n^L, x_{n}^{{x_{n-1}}^{-1}} = x_n^D$ for fixed $L, D \in \Z$.
The following lemma states a formula to collect any word of the form $x_n^i x_{n-1}^j$.

    \begin{lemma} For any $A, B \in \Z$, we have the formula
   \begin{equation}\label{swap2}
       x_n^B x_{n-1}^A =  \begin{cases}x_{n-1}^A {x_n}^{BL^A} \ \text{if } \ A \geq 0 \\ x_{n-1}^A{x_n}^{BD^{-A}} \ \text{if } \ A<0 \end{cases}
   \end{equation}

    \end{lemma}
  The following formula allows the computation of the product of $r$ words given in normal form. It is easily verified by induction on $r$.
     
     \begin{lemma} Let $r\geq 1$ and $w_i = x_{n-1}^{A_i} x_n^{B_i}$ for $1\leq i \leq r$, and $w=w_1 w_2 \ldots w_r$.  Define \[k_j = \begin{cases} 1, A_j \geq 0 \\ 0, A_j<0,     \end{cases}, \ A_i^+=\sum\limits_{j=i+1}^{r}A_j k_j, \ A_i^-=\sum\limits_{j=i+1}^{r} A_j (1-k_j).\] Then, $w=x_{n-1}^A x_n^B$, with $A = \sum\limits_{i=1}^r A_i$, and $B = \sum\limits_{i=1}^{r} B_i L^{A_i^+} D^{A_i^-}$  \end{lemma}

    \begin{lemma} Let $K>0, k_j$ as defined in the above Lemma and consider an element $w=x_{n-1}^A x_n^B$. Write $E(A) = L$, if $A \geq 0$, $E= D$ if  $A<0$, $F(A) = L$, if $A < 0$, $F(A)= D$ if  $A \geq 0$ ($E_j = k_jL+(1-k_j)D$). Then  
    $(x_{n-1}^A x_n^B)^{K} = x_{n-1}^{KA} x_n^{B\left(\frac{E^{\card{KA}} - 1}{E^{\card{A}}-1}\right)} $, and $(x_{n-1}^A x_n^B)^{-K} = x_{n-1}^{-KA} x_n^{-BF^{\card{KA}}\left(\frac{E^{\card{KA}} - 1}{E^{\card{A}}-1}\right)} $. \\ 
  \end{lemma}

Thus, operations on normal words in $x_{n-1}$ and $x_n$ have complexity $\mathcal{O}(1)$. We now consider words in the three generators $x_{n-2}, x_{n-1}, x_n$. 

Define constants $A_j^{(k)},B_j^{(k)}$ with $x_{n-2}^{-k} x_j x_{n-2}^k = x_{n-1}^{A_j^{(k)}} x_{n}^{B_j^{(k)}}$ for $j = n-1, n$, and $k\in \Z$. Clearly, $A_j^{(\pm 1)}$ and  $B_j^{(\pm 1)} $ can be read from the group presentation. 

Given $A=A_j^{(k)}$ and $B=B_j^{(k)}$ for any $k\geq 0$, we obtain $A_j^{(k+1)}$ and $B_j^{(k+1)}$ (resp. $A_j^{(k-1)}$ and $B_j^{(k-1)}$ if $k<0$) in time $\mathcal{O}(1)$ by computing $x_{n-2}^{-1} (x_{n-1}^Ax_n^B) x_{n-2} = (x_{n-1}^{x_{n}})^A (x_3^{x_{1}})^B = w_1^A w_2^B$ which requires two substitutions from the presentation, two exponentiations and one word multiplication, of words in $x_{n-1}$ and $x_n$. The square and multiply method can be used for subsequent exponents, giving a total complexity of $\mathcal{O}(\log k)$ for computing $(A_j^{(k)}, B_j^{(k)})$. The complexity of multiplications with three generators is therefore $\mathcal{O}(\log N)$ where $N$ is the exponent of the generator $x_1$.

Note that if $n\geq 4$ the complexity of computing $x_{n-3}^{-1} (x_{n-2}^A x_{n-1}^B x_n^C) x_{n-3}$ depends on all of the exponents $A, B, C$, and their intermediate values, and from this case onwards nothing concrete can be said about the complexity in general.

  Below, we consider the case where we have a general bound $N$ for the exponents on each generator at each step of collection. For instance, if each generator $a_i$ has finite order $m_i$, the exponents on $a_i$ can always be reduced in polynomial time, so we can assume without loss of generality that $N= \max(m_i)$. In particular this holds for all finite polycyclic groups.

\begin{proposition}Let $n\geq 4$. For $j \leq n-3$, the complexity of multiplying two words, and of inverting a single word, in the generators $x_j$, $x_{j+1}, \ldots, x_n$ is  $\mathcal{O}(\frac{(n-j)!}{2} (\log N)^{2(n-j-2)+1})$. \end{proposition} 
\begin{proof}
Note that from the above discussion we have a complexity of $\mathcal{O}(\log N)$ for multiplication and inversion of words in $x_{n-2}$, $x_{n-1}$ and $x_n$.
For $j \leq n-3$ we have \begin{align*}
    x_{j}^{-1} \; (x_{j+1}^{A_{j+1}} \; x_{j+2}^{A_{j+2}} \; \ldots \; x_n^{A_n} )x_{j} 
    =& {(x_{j+1}^{x_j})}^{A_{j+1}} \; {(x_{j+2}^{x_j})}^{A_{j+2}} \; \ldots \; {(x_{j+1}^{x_j})}^{A_{n}} \\
    =& w_{j+1, j}^{A_{j+1}} \; w_{j+2, j}^{A_{j+2}} \; \ldots \; w_{n, j}^{A_{n}} \\ 
    =& \overline{w_{j+1, j}} \; \overline{w_{j+2, j}} \; \ldots \; \overline{w_{n, j}} \\ 
    =& w 
\end{align*} where the expression in the second line is obtained through $(n-j)$ substitutions from the presentation, the expression in the third line is computed in $\sum_{i=j+1}^{n} O(\log A_i)$ word multiplications in $x_{j+1}, \ldots, x_n$, and the final value $w$ is computed in $(n-j)$ word multiplications in $x_{j+1}, \ldots, x_n$. So, $x_{j}^{-1} (x_{j+1}^{A_{j+1}} x_{j+2}^{A_{j+2}} \ldots x_n^{A_n} )x_{j}$ can be computed in $\mathcal{O}((n-j)\log N)$ word multiplications in $x_{j+1}, \ldots, x_n$, and so  $x_{j}^{-K} (x_{j+1}^{A_{j+1}} x_{j+2}^{A_{j+2}} \ldots x_n^{A_n} )x_{j}^K$ can be computed in $\mathcal{O}( (n-j)\log K \log N)$ word multiplications in $x_{j+1}, \ldots, x_n$.
Thus, two normal words in $x_j$, $x_{j+1}, \ldots, x_n$ can be multiplied by plugging in the value of the conjugation by a power of $x_j$ and then performing a multiplication of two words in $x_{j+1}, \ldots, x_n$. So, the total complexity is $\mathcal{O}((n-j)(\log N)^2)$ word multiplications in $x_{j+1}, \ldots, x_n$. The result on multiplication then easily follows by backwards induction on $j\leq n-3$.
Similarly, a normal word in $x_j$, $x_{j+1}, \ldots, x_n$ can be inverted by  performing an inversion of a word in $x_{j+1}, \ldots, x_n$ and then plugging in the value of the conjugation by a power of $x_j$. So, one must perform $\mathcal{O}((n-j)(\log N)^2)$ word multiplications in $x_{j+1}, \ldots, x_n$ and one inversion of a word in $x_{j+1}, \ldots, x_n$. It may be easily verified that this inversion too has an overall complexity of $\mathcal{O}(\frac{(n-j)!}{2} (\log N)^{2(n-j-2)+1})$.
\end{proof}

\subsection{Analysis of the Conjugacy Search Problem}
\subsubsection{CSP in a Polycyclic Group with two generators}\label{pc-csp}

We consider the case $n=2$, with two generators $x_1$ and $x_2$. Throughout, we will write $N_1 = \ord(x_1)$ and $N_2 = \ord(x_2)$ as the respective orders of $x_1$ and $x_2$ in $G$, which are both allowed to be infinite. We have two relations $x_1^{-1} x_2 x_1 = x_2^{L}$ and $x_1 x_2 x_1^{-1} = x_2^{D}$ (the second is redundant if and only if $N_1$ is finite, in which case $D = L^{N_1-1}$). Note that if $N_2$ is finite then $\gcd(L,N_2) = 1$, since if not, writing $L_2=\gcd(L,N_2)\neq 1$, we have $x_1^{-1}x_2^{N_2/L_2}x_1 = 1$, or $x_2^{N_2/L_2}=1$, a contradiction.

The following lemma is a consequence of Lemma \ref{swap2} for the solution of the CSP.

  \begin{lemma}\label{pc-main} The conjugated word
         $(x_1^cx_2^d)^{-1}(x_1^a x_2^b) (x_1^cx_2^d )$ can be collected to $x_1^g x_2^h$ with $ g =a$ and 
          \[h= \begin{cases}-dL^a+bL^c+d; \ \text{if} \ c,\ a \geq 0  \\
        -dL^a+bD^{-c}+d; \ \text{if} \ c<0,\ a\geq 0 \\
        -dD^{-a}+bL^c+d; \ \text{if} \ c \geq 0, \ a<0 \\
        -dD^{-a} + bD^{-c}+d; \ \text{if} \ c,\ a <0
        \end{cases} \]
    \end{lemma}
  
 \begin{theorem}
 If $N_2=\ord(x_2)$ is finite, the CSP has a polynomial time solution in $G_2$.    \end{theorem}     \begin{proof}
    Suppose we are given an instance of the CSP, i.e. an  equation $(x_1^cx_2^d)^{-1}(x_1^a x_2^b) (x_1^cx_2^d )= x_1^e x_2^f$, where we want to solve the for $c$ and $d$. Then, from Lemma \ref{pc-main}, $a= e \pmod {N_1}$ and the CSP is reduced to solving a modular equation for two unknowns $c$ and $d$. 
    
    If $a\geq 0$, we have $f+d(L^a-1) = bL^c$, or $f+d(L^a-1) = bD^{-c}$. Writing $b_1 = \gcd(b,N_2)$, we see that a solution for $L^c$ (resp $D^{-c}$) exists if and only if $ d(L^a-1) = -f \pmod {b_1}$. Writing $a_1=\gcd(b_1, L^a-1)$, a solution $d$ for $ d(L^a-1) = -f \pmod{b_1}$ exists if and only if $a_1 \mid f$.  By construction, a solution $(c,d)$ exists, so both these conditions are satisfied. Further, a solution $d$ to $ d(L^a-1) = -f \pmod {b_1}$ is given by $d  = -(f/a_1)((L_a-1)/a_1)^{-1} \pmod {b_1/a_1}.$ Write $d = -(f/a_1)((L_a-1)/a_1)^{-1} +Mb_1/a_1$ for some $M\in \Z$ which we may choose. Then, 
\[M(L^a-1)/a_1= (f+d(L^a-1))/b_1= 
\begin{cases} (b/b_1)L^c, \; c \geq 0 \\
(b/b_1)D^{-c}, \; c<0
\end{cases}\]
Writing $A = (b/b_1)^{-1}((L^a-1))/a_1$ (clearly $\gcd(A,N_2)=1$), we may take $M=A^{-1} \pmod {N_2}$, so that a solution is given by $c=0$. Then $d= (L^a-1/a_1)^{-1}(-f+b)/a_1)).$
    
Similarly, a solution can be obtained for the case $a<0$ when $N_1 = \infty$. Thus, in both cases, a solution of the CSP involves a fixed number of applications of the Euclidean algorithm, and so has polynomial time complexity.
    \end{proof}

\begin{theorem}
If $N_2=\ord(x_2)$ is finite, the $\langle x_1 \rangle$-restricted CSP in $G_2$ reduces to a DLP. Further, the elements can be chosen so that it is exactly equivalent to a DLP in $(\Z/N_2\Z)^*$. \end{theorem} \begin{proof}
Here we have the exponents from Lemma \ref{pc-main}, with $d=0$, so the CSP reduces to the retrieval of the exponent $c$ where $f= bL^c \pmod {N_2}$. Here, for a solution to exist, $b_1 = \gcd(b, N_2)$ divides $f$, and we have $f/b_1 = L^c \pmod {N_2/b_1} $, so the adversary must solve the DLP with base $L$ $\pmod {N_2/b_1}$, for the exponent $c$. Choosing the base element so that $b$ satisfies $\gcd(b,N_2)=1$, the restricted CSP is exactly equivalent to a DLP in $(\Z/N_2\Z)^*$.
\end{proof}

\begin{remark}\label{inf-rem}
 If $N_2 = \infty$, then the CSP in $G_2$ reduces to an exponential Diophantine integer equation $f=-dL^a+bL^c+d$. As far as the authors' knowledge goes, there is no known standard technique for solving such equations, and trial and error would perhaps be the best method (for a general reference see \cite{shorey_tijdeman_1986}). For instance, the adversary may try different values of $c$ until $f -bL^c$ is a multiple of $L^a-1$, and subsequently solve for $d$. On the other hand, the $\langle x_1 \rangle$-restricted CSP in this case has an easy solution: the adversary solves $f=bL^c$ for $c$ simply by taking the real number base-$L$ logarithm of $f/b \in \Z$.
\end{remark}

\subsubsection{CSPs in some other polycyclic groups}

\paragraph{$\langle x_1 \rangle$ -restricted CSP in a polycyclic group with three generators}

Here it will be convenient to write $s=x_1$, $t_1=x_2$, $t_2=x_3$ in the group presentation \eqref{pcpres}. Also, write $S= \langle s\rangle $ and $T= \langle t_1, t_2 \rangle$,  $\theta=\ord(s) $, $\theta_1 = \ord(t_1)$, $\theta_2 = \ord(t_2)$ as the respective orders in $G$. Then $T$ is a polycyclic group with two generators, and so from Lemma \ref{swap2} we have $t_2^B t_1^A = t_1^A t_2^{BL^A}, \ A, B \in \Z$. Further, $S $ acts on $T$  via $s^{-1} t_1 s = t_1^{a_1^{(1)}}t_2^{a_2^{(1)}}, \    s^{-1} t_2 s = t_1^{a_1^{(2)}}t_2^{a_2^{(2)}} $ for fixed integers $a_i^{(j)}, \ i, j\in \{1,2 \}$.

Representing an element $t_1^A t_2^B$ of $T$ as a tuple $(A,B) \in  \Z/\theta_1\Z \times \Z/\theta_2\Z$ and writing $(t_1^A t_2^B)^{s^i}=t_1^{A_i}t_2^{B_i}$ we can describe the action of $S$ as
a recurrence relation given by $(A_0, B_0) = (A,B)$, \begin{small}\begin{equation*} 
     (A_{i+1}, B_{i+1}) = \left(a_1^{(1)}A_i + a_1^{(2)}B_i \pmod{\theta_1}, \ a_2^{(1)}L^{A_i a_1^{(2)}} \frac{L^{A_ia_1^{(1)}}-1}{L^{a_1^{(1)}}-1} + a_2^{(2)}\frac{L^{B_ia_1^{(2)}}-1}{L^{a_1^{(2)}}-1} \pmod{\theta_2} \right) 
 \end{equation*}\end{small}
Note that $A_i$ is always given $\pmod {\theta_1}$ and $B_i$ is always given $\pmod {\theta_2}$. While the computation of $A_{i+1}$ and $B_{i+1}$ involves a ``coupling" between these values, the final values seen are always reduced, since they are the exponents in a reduced form word expression. 
Further, while the first component $A_{i+1}$ of the tuple $(A_{i+1}, B_{i+1})$ is linear in $A_i$ and $B_i$, it is no longer linear in the previous terms of the sequence, since $B_i$'s relationship to $A_{i-1}$ and $B_{i-1}$ is non-linear. The general complexity of the $\langle s \rangle$-restricted DLP, i.e. recovering $i$ modulo $\theta$ from $(A_i, B_i)$ is not clear. However, note that when $a_1^{(2)}=0=a_2^{(1)}$, the problem reduces to a DLP in $(\Z/\theta_1\Z)^*$.

\paragraph{$\langle x_1 \rangle$-restricted CSP in $n$ generators when $\langle x_2, \ldots, x_{n} \rangle$ is abelian }

Here it will be convenient to write $s=x_1$, $t_1=x_2, \ldots, t_{n-1}=x_{n}$ in the group presentation \eqref{pcpres}. Also write $S= \langle s\rangle $ and $T= \langle t_1, \ldots, t_{n-1} \rangle$,  $\theta=\ord(s) $, $\theta_i = \ord(t_i)$. Write $t_i^s = t_1^{a_1^{(i)}} \ldots t_{n-1}^{a_n^{(i)}}$ for $1 \leq i \leq n$. Representing the elements of $T$ as column vectors $(r_1 \ldots, r_{n-1})$ of the exponents $r_i$ of the $t_i$, we can describe the action of $s$ on $T$ by the endomorphism \[\Z_{\theta_1} \times \Z_{\theta_2} \times \ldots \times \Z_{\theta_n} \rightarrow \Z_{\theta_1} \times \Z_{\theta_2} \times \ldots \times \Z_{\theta_n} \] \[ (r_1, \ldots, r_n) \rightarrow \begin{bmatrix}
a_1^{(1)} & \ldots & a_1^{(n)}\\
a_2^{(1)} & \ldots & a_2^{(n)} \\
\vdots & \cdots & \vdots \\
a_n^{(1)} & \ldots & a_n^{(n)} \\
\end{bmatrix} \cdot \begin{bmatrix} r_1 \\ r_2 \\ \vdots \\ r_n \end{bmatrix}.
\]
Note that since $a_i^{(j)}$ is always given modulo $\theta_i$ the entries of each column in the above matrix (call it $M$) actually lie in separate groups. However, we may obtain a well-defined matrix power $M^N$ by first computing the power over the integers, and then reducing the $i^{th}$ column modulo $o_i$. Then, the action of $s^N$ on $T$ is given by the endomorphism described by the matrix $M^N$ and the $\langle x_1 \rangle$-CSP is simply the problem of recovering $N$ given $M^N$. Note that this is in general not the same as a matrix DLP because the entries of each column actually lie in separate groups. However, if $ \theta_1 = \ldots =  \theta_m$ and $M$ is invertible over $\Z / \theta_1\Z$, we can obtain $N$ by solving the matrix DLP in the subgroup $\langle M \rangle$ of the ring $\Mat(\Z / \theta_1\Z)$.

\subsection{Examples}\label{pc-ex}
\subsubsection{Holomorphs of Squarefree Cyclic Groups}

The holomorph of a group $H$ is defined as the natural semidirect product of $H$ with its automorphism group $\Aut(H)$. Let $\Hol(C_p) = C_p \rtimes \Aut(C_p)$ of a cyclic group $C_p$ of prime order $p$, generated by $g$. $\Aut(C_{p})\cong \Z_{p}^\times $ is cyclic so $\Hol(C_p)$ is a polycyclic group with two generators. The action of $\Z_{p}^\times$ on $C_p$ is written as a conjugation: $k^{-1} h^i k  = h^{ik}, \ k \in \Z_{p}^\times, \ i\in \Z$.

Given conjugate elements $ g_1= h^l{k_1}$, $g_2 = h^n k_2$ in $G$, suppose we want to find $g = h^m k$ such that $g^{-1}g_1 g = g_2$. It is easy to check that $(h^m k)^{-1}(h^l{k_1})(h^m k)= h^W (k^{-1} k_1 k)$ with $W= k((-m+l) + m k_1^{-1})$. Subsequently we get $h^W = h^n$ and $ (k^{-1} k_1 k) =k_2$. The latter equation is trivial, so one only needs to solve $n= k(({k_1}^{-1}-1)m+l)\pmod p$
for $m$ and $k$. 

If $k_1 \neq 1 \pmod p$ then $k_1^{-1}-1 \in \Z_{p}^\times$ and for any $k$ we find $m=\frac{1}{{k_1}^{-1}-1} (\frac{n}{k} - l )$. If $k_1 =1 \pmod p$ then $lk= n \pmod p$. If $l=0 \pmod p$ then $n=0 \pmod p$ and both $m$ and $k$ take any value (here $g_1=g_2=1$). If $l\neq 0 \pmod p$ then $k=l^{-1}n$ and $m$ takes any value. 

We remark here that $\Hol(C_p)$ embeds into the ring $\Mat_2(\F_p)$ of $2\times 2$ matrices over $\F_p$, so the above solution has an equivalent matrix formulation. This solution may also easily be generalized to abelian groups of squarefree order, for which there is a direct decomposition into cyclic factors, computable in polynomial time \cite{li-bin}.

\subsubsection{Generalized Quaternions}

A generalized quaternion group is a finite polycyclic group given by the presentation
\begin{equation}\label{quatpres}
    Q_{2^n} = \langle x, y \mid x^N = 1, y^2 = x^{N/2}, yx = x^{-1}y, N= 2^{n-1}\rangle. 
\end{equation}
Clearly, any element in this group has a normal form $x^{i}y^j$, where $ 0\leq i \leq N $, $0\leq j \leq 1$. 

One easily derives the relation ${y^j}{x^i} = x^{i(-1)^j}y^j $ for all $i,j \in \Z$. Suppose we have a CSP instance $(x^i y)^{-1} (x^a y) (x^i y) = x^A y$ and want to solve for $i$. We have, $(y^{-1} x^{-i}) (x^a y) (x^i y )= y^{-1} x^{a-2i} y^2 =  x^{2i-a}  y$. Thus, the exponent $i$ is found by solving  $2i -a = A \pmod N$. Note that $(x^i y)^{-1} (x^a) (x^i y) = y^{-1} x^a y = x^{-a} $, so in this case any value of $i$ is a valid solution. Similarly, since $x$ lies in the center of $Q_{2^n}$, the $\langle x\rangle$-restricted CSP is trivial. 

 \subsection{Using matrix representations of polycyclic groups}\label{pc-mat}

Let $G$ be given by the presentation \eqref{pcpres}.  It is known that every polycyclic group is linear, and thus embeds faithfully into a matrix group over some field. More precisely, there exists $m>0$, a field $\F$ and an injective homomorphism $\phi: G \rightarrow \GL_m(\F)$. 

Suppose that this matrix representation is used by the designer of the cryptosystem to hide the structure of $G$. The public parameters are the generator matrices $M_i = \phi(a_i)$ and a base matrix $M_x$, and operations take place in $ \GL_m(\F)$. Let $M_y$ denote one of the conjugated public keys. Note that given the generator matrices, 
to reconstruct the presentation of $G$ and reduce the problem back to the CSP in $G$, an adversary is faced with the following problem: given a matrix $X \in \langle  M_1, \ldots, M_n \rangle$, find integers $(i_1, \ldots, i_n)$ such that $X= M_1^{i_1} \ldots M_n^{i_n}$. Clearly, by solving $\mathcal{O}(n^2)$ instances of this problem the adversary can compute the presentation of $G$ and the words representing $M_x$ and $M_y$, thereby reducing the problem back to the CSP in $G$.

This problem has been discussed in \cite{klingler2009discrete}, and is called the Generalized Discrete Logarithm Problem (GDLP). The thesis \cite{ilic} discusses some square-root type algorithms for the GDLP in finite matrix groups. The case $n=2$ has been discussed for general finite groups in \cite{meshram} and \cite{kashyap}, both of which show a square root algorithm to reduce the GDLP to at most two DLPs. Observe that in $Q_{2^n}$, this process introduces a single matrix DLP into the protocol:  an adversary sees matrices $A=M_x^i M_y$ or $A=M_x^i$, and so can recover $i$ by solving one of the matrix DLPs $AM_y^{-1} = M_x^i, \ A = M_x^i$. Therefore, in general using the matrix representation likely does not offer any novel security feature, though it may enhance the overall security.

Now, suppose that the original problem is given as a CSP in $G$ and the adversary is able to efficiently compute a faithful representation $\phi: G\rightarrow \GL_m(\fq)$ as well as its inverse. Denote by $x \in G$ the public base element and $y=g^{-1}xg \in G$ the public key. Then, it suffices for the adversary to find a matrix $M_g \in \phi(G)$ such that $M_g^{-1} \phi(x) M_g = \phi(y)$, so solving a CSP in $\phi(G)$ breaks the system. In fact, if the original CSP in $G$ is an $A$-restricted CSP for $A \leq G$ cyclic, then it is not even necessary to compute $\phi^{-1}$ since the secret here is essentially an integer. In Theorem \ref{mat-main}, we will see that an $A$-restricted CSP in $\GL_m(\F)$ is reducible to a set of $\mathcal{O}(m^2) $ DLPs over $\F$. So, any system with a linear platform must ensure that the subgroup from which conjugators are chosen has at least two generators. We remark here that in the case where such an efficient representation $\phi$ and its inverse are available, several attacks exist to directly retrieve the shared key, see \cite{tsaban2015polynomial}, \cite{kreuzer2014linear}, \cite{myasnikov2015linear},  \cite{ben2018cryptanalysis} for details.

\section{Matrix Groups}\label{mat-sec}

Throughout this section, $q$ denotes a power of a prime $p$ and $\fq$ denotes the finite field with $q$ elements, and $\ord(a)$ denotes the multiplicative order of an element $a\in \mathbb{F}_q^*$. 

Matrix groups over finite fields have served as platform groups for several proposed protocols. In \cite{Menezes1997} and \cite{Freeman2004}, the DLP over the matrix group $\GL_n(\F_q)$ was studied and shown to be no more difficult than the DLP over a small extension of $\F_q$, and consequently less efficient in terms of key sizes for the same security level. Most known nonabelian platform groups are linear, i.e. they embed faithfully into a matrix group. If this embedding and its inverse can efficiently be computed by an adversary, the security of the system depends on that of the matrix CSP rather than that in the original platform. 

It is then natural and important to study the complexity of the CSP over matrix groups. Several attacks exist to directly retrieve the shared key from CSP-based protocols without computing the secret keys \cite{tsaban2015polynomial}, \cite{kreuzer2014linear}, \cite{myasnikov2015linear},   \cite{ben2018cryptanalysis}. However, to the best of our knowledge, the CSP and its variants have not been investigated.
In this section, we study the $A$-restricted matrix CSP
for a cyclic subgroup $A \subseteq \GL_n(\F_q) $. Here, for maximum generality we also allow the base element be a non-invertible matrix. In other words, we provide a cryptanalysis of Protocol~\ref{mat-prot} over a ring $R$, for the case $R =\Mat_n(\F_q) $ of $n\times n$ matrices over $\fq$. 

\begin{protocol}\label{mat-prot}
Let $X \in R$ and $Z \in R^{*}$ be public elements.
\begin{enumerate}
    \item Alice picks a secret integer $r\in \Z$ and publishes $Z^{-r} X Z^r$.
  \item Bob picks a secret integer $s \in \Z$ and publishes $Z^{-s} X Z^{s}$.
    \item The shared secret is $Z^{-s-r} X Z^{s+r}$.
\end{enumerate}
\end{protocol}

In subsection 1, we provide a polynomial time reduction to recover $r \in \Z$ to a set of $\mathcal{O}(n^2)$ DLPs. In subsection 2, we show how this enables a full cryptanalysis of the system proposed in \cite{quat}.

\begin{remark}\label{gen-crt}
Similarly to \cite{Freeman2004} (also mentioned in \cite{Menezes1997}) we will use an easy generalization of the Chinese Remainder Theorem (CRT) for systems of equations of the form $x \equiv x_i \pmod {\theta_i}, 1\leq i\leq s$, where the $\theta_i$ are not necessarily coprime, but a solution is required $\pmod \theta:=\lcm_i \theta_i$. Write $\theta=p_1^{e_1} \ldots p_t^{e_t}$ as the prime factorization of $\theta$. For each $j \in \{1,\ldots, t\}$, let $\emptyset \neq \mathcal{I}_j\subseteq \{1,\ldots, s\}$ denote the list of indices $i$ such that $p_j^{e_j} \mid \theta_i$. If a solution $x$ to the original system of equations exists we have $x_i-x_{i'}= 0 \pmod {p_j^{e_j}}$ for each $i,i' \in \mathcal{I}_J$ and $x=x_{i_0} \pmod {p_j^{e_j}}$ for any $i_0 \in \mathcal{I}_j$. So, if a solution exists, we can translate the original system to one where the moduli are coprime, and so can be solved with the CRT.
\end{remark}

\subsection{$\langle Z\rangle$-restricted CSP in $\GL_n(\fq)$}\label{mat-csp}

Suppose that $Z^{-r} X Z^r =Y$, the adversary sees $X$, $Z$, and $Y \in \Mat_n(\fq)$, the integer $r$ is secret. There exists an extension $\F_{q^k}$ of $\F_{q}$ and a unique matrix $P \in \GL_n(\F_{q^k})$ (computable in polynomial time, by Algorithm 1 in \cite{Menezes1997}) such that $J_Z = PZP^{-1}$, where  $J_Z$ is the Jordan Normal form of $Z$.
Here $\F_{q^k}$ is the smallest  extension containing all eigenvalues of $Z$, and $k$ is polynomial in $q$ and $n$ \cite{Menezes1997}. Let $\theta_Z$ be the order of $Z$ in the group $\GL_n(\fq)$ and $\theta_J$ be the order of $J_Z$ in the group $\GL_n(\F_{q^k})$. Then, clearly since $J_Z = PZP^{-1}$ we have $\theta_Z = \theta_J$.  Further, we only require a value of $r$ modulo  $\theta_Z$ to break the system.
 
Consider $M=P X P^{-1}$ and $N=PYP^{-1}$. Note that these are both computable by the adversary. It is easily verified that $    Z^{-r} X Z^r =Y \iff J_Z^{-r} M J_Z^r = N.$ Thus, to recover $r$ we may assume that $Z$ is already in Jordan form. We divide the rest of the analysis into two cases. We first consider the case where $Z$ is diagonalizable. 

 \subsubsection{Case: $J_Z$ is diagonal.}
 
 We write $M=(M_{i,j})_{n\times n}$ and $N=(N_{i,j})_{n\times n}$. 
 
 \begin{theorem}\label{diag}
 If $J_Z=
\begin{pmatrix}d_1 & \ldots & 0\\
\vdots  & \ddots & \vdots \\
0  & \ldots & d_n
\end{pmatrix}$ is diagonal then
the retrieval of $r$ in Protocol~\ref{mat-prot} reduces to solving at most $n^2$ DLPs over $\F_{q^k}$.
 \end{theorem}
 \begin{proof}
 Note that $d_i \neq 0 \ \forall \ i$ since $J_Z$ is invertible.   
We expand $J_Z^{-r} M J_Z^r = N $: 
\begin{align} 
& \begin{pmatrix}d_1^r &  \ldots & 0\\
\vdots  & \ddots & \vdots \\
0 & \ldots &d_n^r
\end{pmatrix}  \begin{pmatrix}M_{11} & \ldots & M_{1n}\\
\vdots &  \ddots & \vdots \\
M_{n1} & \ldots & M_{nn}
\end{pmatrix} \begin{pmatrix}d_1^{-r} & \ldots & 0\\
\vdots & \ddots & \vdots \\
0 & \ldots & d_n^{-r}
\end{pmatrix}= \begin{pmatrix}N_{11}  & \ldots & N_{1n}\\
\vdots  & \ddots & \vdots \\
N_{n1} & \ldots & N_{nn}
\end{pmatrix} \nonumber\\ 
\iff & M_{ij}(d_i d_j^{-1})^r = N_{ij}, \  1\leq i, j \leq n. 
\end{align}

By assumption $M $ and $N$ are nonzero matrices, so there exists at least one pair of indices $(i,j)$ such that $M_{ij}\neq 0$, $N_{ij}\neq 0$. For any such pair we have $(d_{i}d_j^{-1})^{r}= M_{ij}^{-1}N_{ij}$, so $r$ $\pmod{\ord(d_i d_j^{-1})}$ is found by solving a DLP in $\F_{q^k}$. Repeating this for all such pairs $(i,j)$, we may compute $r$ modulo ${\lcm\limits_{\exists (i,j) \mid M_{ij}\neq 0}(\ord(d_i^{-1}d_j)) }$ using the generalized Chinese Remainder Theorem (see Remark \ref{gen-crt}). Clearly, this value of $r$ satisfies the equation, as required.
 \end{proof}

\subsubsection{Case: $J_Z$ is not diagonal}

Suppose that $\small  J_Z =\begin{pmatrix}J_1 & \ldots & 0\\
\vdots & \ddots  & \vdots \\
0 & \ldots & J_s
\end{pmatrix}$ is the Jordan-Normal form of $Z$, where each $ \small J_i =\begin{pmatrix}\lambda_i & 1&\ldots & 0\\
\vdots  & \vdots & \ddots & \vdots \\
0   & 0 & \ldots & 1
 \\
0   & 0 & \ldots & \lambda_i
\end{pmatrix}, \ 1 \leq i \leq s$, is a $d_i \times d_i$ Jordan block corresponding to the eigenvalue $\lambda_i\in \F_{q^k}$, and $d_i >1 $ for at least one $i$. Denoting $\binom{-n}{k}:=(-1)^k\binom{n+k-1}{k}$, and with the convention $\binom{r}{m}=0$ if $m<r$ and $r>0$, we have by induction, for $r\geq 1$, \begin{align*}
    \small J_i^r=\begin{pmatrix} \lambda_i^{r} & \binom{r}{1} \lambda_i^{r-1} &\ldots & \binom{r}{d_i-1}\lambda_i^{r-d_i+1}\\
0&\lambda_i^{r}  &\ldots & \binom{r}{d_i-2}\lambda_i^{r-d_i+2}\\
\vdots &  \vdots & \ddots & \vdots \\
0  & 0 & \ldots & \binom{r}{1}\lambda_i^{r-1} \\
0 &0  & \ldots & \lambda_i^{r} \end{pmatrix}, \quad J_i^{-r}=\begin{pmatrix} \lambda_i^{-r} & \binom{-r}{1} \lambda_i^{-r-1} &\ldots & \binom{-r}{d_i-1}\lambda_i^{-r-d_i+1}\\
0&\lambda_i^{-r}  &\ldots & \binom{-r}{d_i-2}\lambda_i^{-r-d_i+2}\\
\vdots &  \vdots & \ddots & \vdots \\
0  & 0 & \ldots & \binom{-r}{1}\lambda_i^{-r-1} \\
0 &0  & \ldots & \lambda_i^{-r}
\end{pmatrix}
\end{align*}
 More concisely, for any $1\leq i \leq s$, $r\in \Z$, $({J_i}^r)_{(k,l)}=\binom{r}{l-k}\lambda_i^{r-l+k}$, $0\leq k,l\leq d_i$. Now, we write $M$ and $N$ as block matrices with $s^2$ $d_i \times d_j$ blocks $(M_{i,j})_{d_i \times d_j}$, $ (N_{i,j})_{d_i \times d_j}$:
\[M = \begin{pmatrix}(M_{1,1})_{d_1 \times d_1}&  \ldots & (M_{1,s})_{d_1 \times d_s}\\
\vdots  & \ddots &   \vdots \\
(M_{s,1})_{d_s \times d_1}  &  \ldots &  (M_{s,s})_{d_s \times d_s} \end{pmatrix}, \quad N = \begin{pmatrix}(N_{1,1})_{d_1 \times d_1} &   \ldots & (N_{1,s})_{d_1 \times d_s}\\
\vdots  & \ddots &   \vdots \\
(N_{s,1})_{d_s \times d_1} &  \ldots &  (N_{s,s})_{d_s \times d_s} \end{pmatrix}.\]

The next result reduces the problem of recovering $r$ from Protocol~\ref{mat-prot} to a set of $s^2$ matrix equations involving Jordan blocks.
\begin{lemma}\label{block}
  $J_Z^{-r} M J_Z^r = N \iff J_i^{-r}M_{i,j}J_j^r = (N_{i,j})_{d_i \times d_j}$ $\forall \ 1 \leq i,j \leq s$.
\end{lemma}
\begin{proof}
 
The result is clear from the observation that $J_Z^{-r} M J_Z^r =$
    \begin{align*}\small
        &\begin{pmatrix}J_1^{-r} & \ldots & 0\\
\vdots  & \ddots & \vdots \\
0 & \ldots & J_s^{-r}
\end{pmatrix} \begin{pmatrix}M_{11} &   \ldots & M_{1s}\\
\vdots  & \ddots &   \vdots \\
M_{s,1}  &  \ldots &  M_{s,s} \\\end{pmatrix}  \begin{pmatrix}J_1^{r} & \ldots & 0\\
\vdots  & \ddots & \vdots \\
0 & \ldots & J_s^{r}
\end{pmatrix} =\begin{pmatrix}J_1^{-r}M_{11}J_1^r & \ldots &  J_1^{-r}M_{1s}J_s^r\\ \vdots  & \ddots & \vdots\\J_s^{-r}M_{s1}J_1^r & \ldots &  J_s^{-r}M_{ss}J_s^r \end{pmatrix}.
    \end{align*}
\end{proof}

 \begin{lemma}\label{formula} Writing $M_{i,j}=(m_{l,k}^{(i,j)})$, $N_{i,j}=(n_{l,k}^{(i,j)})$, we have $J_Z^{-r} M J_Z^r = N \iff$ \[\small n_{f,h}^{(i,j)}=\sum_{k=1}^h \sum_{l=f}^{d_i}\binom{-r}{l-f}\binom{r}{h-k}\ m_{l,k}^{(i,j)}\lambda_i^{-r+f-l} \lambda_j^{r-h+k} \ \forall 1\leq f \leq d_i, \ 1\leq h\leq d_j, 1 \leq i, j \leq s.\]
\end{lemma}
\begin{proof}

Recall that $({J_i}^r)_{(k,l)} = \binom{r}{l-k}\lambda_i^{r-l+k)}$ for any integer $r$ and any $1\leq k, l\leq d_i$. 
We compute the $(f,h)$th term of $J_i^{-r}M_{ij}J_j^r$. Note that for $1\leq g\leq d_i $ the $(f,g)$th term of $J_i^{-r}M_{i,j}$ is given by $\sum_{l=f}^{d_i}\binom{-r}{l-f}\lambda_i^{-r+f-l} m_{l,g}^{(i,j)}$. Thus, the $(f,h)$th term of $N_{i,j}=J_i^{-r} M_{i,j} J_j^r$ is \begin{align}\label{eq-iff}
   & n_{f, h}^{(i,j)} = \sum_{g=1}^h \sum_{l=f}^{d_i}\binom{-r}{l-f}\binom{r}{h-g}\ m_{l,g}^{(i,j)}\lambda_i^{-r+f-l} \lambda_j^{r-h+g}.
\end{align}
For the equality $J_i^{-r}M_{ij}J_j^r = N_{ij}$ we require $n_{f, h}^{(i,j)} = m_{f, h}^{(i,j)}$ for all indices $1 \leq f, h \leq d_i$. The result now follows from Lemma \ref{block}.\end{proof}

Now, by assumption, $M \neq 0$. So, we may choose the largest index $l_0^{(i,j)}$ such that the $l_0^{(i,j)}$th row of $M_{i,j}$ has at least one nonzero term. Let $g_{0}^{(i,j)}$ be the smallest column index such that $m_{l_0,g_0}^{(i,j)} \neq 0$. We have, $m_{l,g}^{(i,j)} = 0 \ \forall \ l>l_0, 1\leq g\leq d_j$, and $m_{l_0, g}^{(i,j)} = 0 \ \forall \ g<g_0$.

Taking $f = l_0, h=g_0+1$, by equation~\eqref{eq-iff} the $(f,h)^{th}$ term of $N_{i,j}=J_Z^{-r} M_{i,j} J_Z^r$ is given by \begin{align}\label{equality}
   n_{l_0,g_0+1}^{(i,j)}= & \sum_{g=1}^{g_0+1}\sum_{l=l_0}^{d_i} \binom{-r}{l-l_0}\binom{r}{g_0+1-g}\ m_{l, g}^{(i,j)}\lambda_i^{-r+l_0-l} \lambda_j^{r-g_0-1+g} \nonumber \\ 
   & = \sum_{g=g_0}^{g_0+1} \binom{r}{g_0+1-g} m_{l_0, g}^{(i,j)}\lambda_i^{-r} \lambda_j^{r-g_0-1-g} \nonumber \\
    & = r m_{l_0,g_0}^{(i,j)} \lambda_j (\lambda_i^{-1} \lambda_j)^r +  m_{l_0,g_0+1}^{(i,j)}(\lambda_i^{-1} \lambda_j)^r
\end{align}

In particular, for $i=j$, we have the equation $rm_{l_0,g_0}^{(i,j)} \lambda_j^{-1} + m_{l_0, g_0+1}^{(i,j)}  = n_{l_0,g_0+1}^{(i,j)} $, where by construction, $m_{l_0 g_0}^{(i,j)}  \neq 0$. This can be solved to get $r':=r \pmod p$. We thus have the following results.

\begin{proposition}\label{modp}
The value of $r':=r \pmod p$ can be computed in polynomial time.
\end{proposition}

\begin{proposition}\label{modlcm}
Computing $r \pmod  {\underset{1\leq i \leq s}{\lcm}\ord(\lambda_i)} $ reduces in polynomial time to solving at most $s^2$ DLPs in $\F_{q^k}$.
\end{proposition}
\begin{proof}
We write $r =r' + pw $ for $w \in \Z$. From equation \eqref{equality}, since $m_{l_0,g_0}^{(i,j)}  \in \F_{q^k}$, we have \[m_{l_0,g_0+1}^{(i,j)}=r m_{l_0,g_0}^{(i,j)}  \lambda_j (\lambda_i^{-1} \lambda_j)^r +  m_{l_0,g_0+1}^{(i,j)} (\lambda_i^{-1} \lambda_j)^r
= (\lambda_i^{-1} \lambda_j)^r(r' m_{l_0,g_0}^{(i,j)}\lambda_j  +m_{l_0,g_0+1}^{(i,j)} ). \] Equating $ (\lambda_i^{-1} \lambda_j)^r(r' m_{l_0,g_0}^{(i,j)}\lambda_j  +m_{l_0,g_0+1}^{(i,j)} ) = n_{l_0,g_0+1}^{(i,j)} $ for $i \neq j$, we see that $r$ can now be recovered $\pmod{ \ord(\lambda_i^{-1} \lambda_j)}$ by solving a DLP. Repeating this process for all pairs $(i,j)$, and using the generalized CRT (see Remark~\ref{gen-crt}) we get $r $ modulo ${ \underset{1\leq i <j \leq s}{\lcm}\ord(\lambda_i^{-1} \lambda_j)= \underset{1\leq i \leq s}{\lcm}\ord(\lambda_i)} $.\end{proof}

\paragraph{Computing $r \pmod{\theta_Z}. $}
As before, $\theta_Z$ denotes the order of $J_Z$ in the group $GL(\F_{q^k})$. We will now show how Propositions \ref{modp} and \ref{modlcm} allow us to compute $r \pmod {\theta_Z}$. 
We have the following result on the value of $\theta_Z$ from \cite{Menezes1997}. 

\begin{lemma}[\cite{Menezes1997}]
The order $\theta_Z$ of $J_Z$ is $\underset{1\leq i \leq s}{\lcm}(\lambda_i)\cdot p\{t\}$, where $t$ is the largest Jordan block in $J_Z$ and $p\{t\}$ denotes the smallest power of p greater than or equal to $t$.
\end{lemma} 

Below, we show that if we can compute $r \pmod p$, then by extension we can find $r \pmod {p\{t\}}$.

\begin{lemma}\label{lifting}
 Suppose that an Algorithm A returns $r \pmod p$ for an equation of the form $Z^{-r} X Z^r=Y$. Then, for any $v \geq 1 $, one may find $r \pmod {p^v}$ with $v$ applications of Algorithm A.
\end{lemma} \begin{proof}
 We write $r = r_0 + r_1p+ \ldots + r_{v-1}p^{v-1} \pmod{p^v}$. With one application of Algorithm A, one finds $r \pmod p = r_0$, and then compute $Z_1=Z^p, \ Y_1=Z^{r_0} Y Z^{-r_0}, \ r'=r_1+r_2p+\ldots + r_{v-1}p^{v-2}$. Then, we have an equation $Z_1^{-r'}XZ_1^{r'} = Y_1$. One now uses Algorithm A again to find $r' \pmod p = r_1$. In $v$ applications of Algorithm A, we recover $r \pmod{p^v}$. 
\end{proof}

Finally, we can prove the final result of this section.

\begin{theorem}\label{nondiag} Let $J_Z$ be non-diagonal, and composed of $s$ Jordan blocks. Then, the computation of $r$ from Protocol~\ref{mat-prot} is polynomial time reducible to a set of $s^2$ DLPs over $\F_{q^k}$.
\end{theorem}

\begin{proof}
 From Propositions \ref{modp} and \ref{modlcm}, we may compute $r \pmod {\underset{1\leq i \leq s}{\lcm}\lambda_i}$ and $r'=r \pmod p$ with the same time complexity as solving $s^2$ DLPs in $\fq$. Note that since the multiplicative order of any element of $\F_{q^k}^\times$ divides $q^k-1$, the values of $r'=r \pmod p$ and $r \pmod {\underset{1\leq i \leq s}{\lcm}\ord(\lambda_i)} $ are obtained independently. Now, write $p\{t\}=p^v\leq q^k$, by Lemma \ref{lifting}, we can obtain $r \pmod {p^v}$ from $r'$ in polynomial time $\mathcal{O}(k\log q)$. Combining $r \pmod {p^v}$ and $r \pmod {\underset{1\leq i \leq s}{\lcm}\lambda_i}$ using the Chinese Remainder Theorem, we get $r \pmod{\theta_Z}$, as required. It is also clear that every step apart from the DLPs has polynomial time complexity. \end{proof}

We may thus summarize the results of Theorems~\ref{diag} and \ref{nondiag} as follows.

\begin{corollary}\label{mat-main}
The $\langle Z \rangle$-restricted CSP in $\GL_n(\fq)$ reduces in polynomial time to $\mathcal{O}(n^2)$ DLPs over a small extension $\F_{q^k}$ of $\fq$. 
\end{corollary}

\begin{remark}
While we have only discussed the search variant of the conjugacy problem, it is not hard to show that the arguments of this section also reduce the decisional variant of the CSP to a set of corresponding decisional versions of the DLP and generalized CRT. Furthermore, this may be used, along with a collision-type square-root algorithm, to solve any $A$-restricted CSP in $\GL_n(\fq)$, where $A$ is abelian. 
\end{remark}

\section{Applications to Cryptanalysis}\label{sec-cryptanalysis}

\subsection{Decomposition over Generalized Quaternions}

\paragraph{CSP in $Q_{2^n}$} Recall the presentation of $Q_{2^n}$ in \ref{quatpres}. One easily derives the relation ${y^j}{x^i} = x^{i(-1)^j}y^j $ for all $i,j \in \Z$. Suppose we have a CSP instance \[(x^i y)^{-1} (x^a y) (x^i y) = x^A y\] and want to solve for $i$. We have, \[(y^{-1} x^{-i}) (x^a y) (x^i y )= y^{-1} x^{a-2i} y^2 =  x^{2i-a}  y.\] Thus, the exponent $i$ is found by solving  $2i -a = A \pmod N$. Note that \[(x^i y)^{-1} (x^a) (x^i y) = y^{-1} x^a y = x^{-a},\] so in this case any value of $i$ is a valid solution. Similarly, since $x$ lies in the center of $Q_{2^n}$, the $\langle x\rangle$-restricted CSP is trivial.

 We now describe a cryptanalysis of Protocol~\ref{compdec1}, which is based on the decomposition problem. 
  \begin{proposition}
Protocol~\ref{compdec1} can be broken in polynomial time by retrieving $a$ and $c$, which reduces to a system of linear equations over $\Z_N$. 

\end{proposition}
  \begin{proof}
  We assume that $a \not\in \langle x \rangle $, since otherwise, finding $a$ is trivial. Similarly, assume $c \not\in \langle x \rangle$.     Write $a = x^A y$, $c=x^C y$, $d_i = x^{D_i}, \ b_i = x^{B_i}, 1\leq i \leq 2$. 
    We have, by collection, \begin{align*}
        x_1 = x^{n_1}y = x^{B_1+A-B_2}y, \quad  x_2=x^{n_2}y =x^{D_1+n_1-D_2}y, \quad x_3 = x^{n_3}y =x^{D_1+A-D_2} y, \\ y_1= x^{m_1}y = x^{D_1+C-D_2}y, \quad
        y_2=x^{m_2}y  =x^{B_1+m_1-B_2}y, \quad
       y_3= x^{m_3}y =x^{B_1+C-B_2}y
    \end{align*}
The adversary sees the $x_i$'s and $y_i$'s, and thus also the $m_i$'s and $n_i$'s, for $1\leq i \leq 3$, and needs to solve linear equations for $A, C$, $D_i$'s, and $B_i$'s. The result is now clear. 
 \end{proof}

Observe that the discussion of this example is applicable in any group of the form $\langle x \rangle \rtimes \langle y \rangle, y^2 \in \langle x \rangle$.
Another notable example is the dihedral group $D_{2n}$.

\subsection{Subgroup Conjugacy Search in Quaternions modulo $p$}
We show how the cryptanalysis for Protocol~\ref{quart} reduces to breaking Protocol~\ref{mat-prot} for $R=\Mat_2(\F_p)$. First, by the arguments in \cite{hurwitz}, we have $H/pH = L/pL = H_p$. Now, by Proposition 3.3 of \cite{hurwitz}, for any integers $a$ and $b$ satisfying $a^2 + b^2 \equiv -1 \pmod p$, the map 
\begin{align*}
    \phi_{a,b}: H/pH & \rightarrow \Mat_2(\Z/p\Z) \\
    \gamma_1 + \gamma_2i + \gamma_3j + \gamma_4 k & \mapsto \begin{pmatrix}
    \gamma_1 + \gamma_2 a + \gamma_4b & \gamma_3 + \gamma_4a - \gamma_2b \\
    -\gamma_3 + \gamma_4a- \gamma_2b & \gamma_1 - \gamma_2a - \gamma_4b.
    \end{pmatrix}
\end{align*}
is an isomorphism of rings. Clearly, the inverse of $\phi_{a,b}$ is also easily computed as

\begin{multline}
\phi_{a,b}^{-1}\left( \begin{pmatrix}
    A & B \\
    C & D
    \end{pmatrix} \right)=\\
    \frac{1}{2}(A+D) + \frac{1}{2}\left(b(B+C)-a(A-D)\right)i + \frac{1}{2}(B-C)j -\frac{1}{2}\left(a(B+C)+b(A-D)\right)k.  
    \end{multline}

Thus, Protocol~\ref{quart} may be treated as if it is over $\Mat_2(\F_p)$. Now, note that the public keys are of the form $y=z^{-r} x^s z^r $ where $r$ and $s$ are private integers, and $x\in \Mat_2(\F_p)$, $z \in \GL_2(\F_p)$ are public matrices. The only difference now with Protocol~\ref{mat-prot} is that this scheme has two secret integers instead of one. However, we observe that the presence of two secret integers weakens the scheme, because an adversary can break the system if they find any pair of solutions $(r,s)$ such that $K_A = z^r x^s x^{-r}$. So, it makes sense to fix $s$, without loss of generality, to $s=1$. Thus, Corollary \ref{mat-main} for the case $n=2$ now shows that finding $r$ and breaking Protocol~\ref{quart} reduces to at most four DLPs.

\subsection{Subgroup Conjugacy Search in Matrix Groups}

We first recall that the SCSP corresponds exactly to the $A$-restricted CSP defined in this paper, for a cyclic subgroup $A$ of $G$. Therefore, Corollary~\ref{mat-main} gives a direct cryptanalysis of Protocol~\ref{mat-prot}, reducing its security to that of a set of $\mathcal{O}(n^2)$ DLPs over a small extension of $\fq$. 

While the authors state that the SCSP is at least as hard as the CSP, we remark that this is likely not true in general. In Section~\ref{pc-csp} we showed that in finite polycyclic group with two generators, a well-chosen SCSP is harder than the CSP, whereas by Remark~\ref{inf-rem} the SCSP in certain infinite polycyclic groups with two generators is seemingly easier than the CSP. Similarly, Section~\ref{mat-csp} gives a complete reduction of the SCSP in $\GL_n(\fq)$ to a set of DLPs, but a solution to the general CSP is unclear. This may be intuitively realized, since while there is an added constraint in the SCSP, the adversary also has more information on where to search for the conjugator.

\section{Conclusion}

In this paper, we discussed the complexity of the conjugacy search problem (CSP) in two important classes of platform groups for nonabelian group-based cryptography, using the protocols in \cite{Gu2014ConjugacySB}, \cite{cmpdec}, and \cite{quat} as our starting point. We produced a polynomial time solution for the CSP in a finite polycyclic group with two generators, and showed that a restricted CSP is reducible to a DLP. In matrix groups over finite fields, we used the Jordan decomposition of a matrix to produce a polynomial time reduction of an $A$-restricted CSP, where $A \subseteq GL(\fq)$ is a cyclic subgroup, to a set of DLPs over an extension of $\fq$. We then used these general methods and results to describe concrete cryptanalysis algorithms for the systems proposed in \cite{Gu2014ConjugacySB}, \cite{cmpdec}, and \cite{quat}. More generally, a direct consequence of our results is that the security of a protocol based on an $A$-restricted CSP, where $A$ is cyclic, which uses a linear platform group, essentially depends on the difficulty of computing a representation of the platform and a set of DLPs. We believe that our methods and findings are likely to allow for several other heuristic attacks in the general case.

\bibliographystyle{plain}
\bibliography{references}

\end{document}